\documentclass[10pt,conference]{IEEEtran}
\usepackage[cmex10]{amsmath}
\usepackage{mathtools}
\usepackage{amssymb}
\usepackage{amsthm}
\usepackage{epsfig}
\usepackage{color}
\usepackage{algorithm}
\usepackage{algorithmic}
\usepackage{multirow}
\usepackage{array}
\newtheorem{theorem}{Theorem}
\newtheorem{definition}{Definition}

\newtheorem{lemma}{Lemma}
\usepackage{epsfig}
\usepackage{color}

\newcommand{\squeezeup}{\vspace{-2.5mm}}

\begin{document}
\title{Jamming aided Generalized Data Attacks:\\ Exposing Vulnerabilities in Secure Estimation}
\author{\IEEEauthorblockN{Deepjyoti Deka,~ Ross Baldick ~and~ Sriram Vishwanath}
\IEEEauthorblockA{Department of Electrical \& Computer Engineering, The University of Texas at Austin\\
Email: deepjyotideka@utexas.edu, baldick@ece.utexas.edu, sriram@ece.utexas.edu }}

\maketitle

\begin{abstract}
Jamming refers to the deletion, corruption or damage of meter measurements that prevents their further usage. This is distinct from adversarial data injection that changes meter readings while preserving their utility in state estimation. This paper presents a generalized attack regime that uses jamming of secure and insecure measurements to greatly expand the scope of common `hidden' and `detectable' data injection attacks in literature. For `hidden' attacks, it is shown that with jamming, the optimal attack is given by the minimum feasible cut in a specific weighted graph. More importantly, for `detectable' data attacks, this paper shows that the entire range of relative costs for adversarial jamming and data injection can be divided into three separate regions, with distinct graph-cut based constructions for the optimal attack. Approximate algorithms for attack design are developed and their performances are demonstrated by simulations on IEEE test cases. Further, it is proved that prevention of such attacks require security of all grid measurements. This work comprehensively quantifies the dual adversarial benefits of jamming: (a) reduced attack cost and (b) increased resilience to secure measurements, that strengthen the potency of data attacks.
\end{abstract}

\section{Introduction}
State Estimation in a vital component for robust control of power system and efficient electricity market operations. It involves collection of measurements from meters distributed across the grid that are communicated through SCADA  (Supervisory Control and Data Acquisition) systems and then using them for determining the system state. Presence of faster sampling meters like phasor measurement units (PMUs) \cite{pmu1}  and Wide-Area Monitoring and Control Systems has led to greater data collection and heightened focus on reliable state estimation. This is because these new meters and their digital communication expose the grid to adversarial data attacks. Adversaries, often cyber in nature, can coordinately change meter readings to produce an incorrect state estimate that can subsequently result in grid failures and sub-optimal electricity prices. In fact, practical adversarial attacks have been widely studied in research (`Aurora' test attack \cite{aurora}, PMU timing attack \cite{todd}) and also reported in national media (cyberspying \cite{wallstreet}, `Dragonfly' virus \cite{dragonfly}). There has thus been a surge in recent research aimed at identifying power grid vulnerabilities and designing resilience to adversarial attacks.

The authors of \cite{hidden} were along the first to identify the problem of `hidden' data attacks that can change the state estimate by bypassing bad-data detection checks at the estimator. The central idea behind `hidden' attacks in \cite{hidden} is the design of a vector of data injections in the column space of the measurement matrix used in state estimation. Different adversarial goals (Eg. minimizing number of compromised measurements, minimum attack energy and cost) and operating conditions (Eg. type of measurements, power flow model, presence of secure measurements) have led to diverse research approaches to the problem of optimal attack construction. For adversaries interested in minimizing the number of measurement corruptions in a DC-power flow based estimator, \cite{poor} uses a $l_0 -l_1$ relaxation based framework to design the optimal `hidden' attack, while \cite{sou} uses mixed integer linear programming. For state estimation relying on voltage phasor and line flow measurements (collected from PMUs), \cite{deka,deka1} provide a graph-cut based `hidden' attack framework. Similarly, \cite{florian} discuss conditions for feasible data attack on a Kalman-Filter based estimator for AC power flow systems. For the related problem of  preventing data attacks, techniques discussed in literature include heuristic scheme \cite{thomas}, greedy schemes \cite{poor,deka1} among others.

Aside from the mentioned research on `hidden' attacks, a recent line of work has analyzed `detectable' data attacks that affect state estimation despite failing bad-data detection checks. An attacker in this case prevents the bad-data remover from removing some/all of the tempered measurements from the system. In this context, reference \cite{frame} demonstrates the construction of a basic `detectable' attack (termed `data integrity' attack) by using half of the measurements in the optimal `hidden' attack, and by damaging the rest. The state estimator here removes only the damaged measurements as bad-data while the other half manipulated by the adversary passes the bad-data detection test and causes the attack. Reference \cite{dekaISGT} generalizes this technique by creating `detectable' attacks from graph cuts that may include a minority of incorruptible measurements. This generalization produces even greater reduction in attack cost (minimum being $50\%$) over `hidden' attack costs. More importantly, the framework in \cite{dekaISGT} produces feasible `detectable' attacks in systems secure against all `hidden' attacks. In this paper, we analyze both attack regimes: `hidden' and `detectable' for adversaries that have an additional tool: measurement jamming.

By jamming, we refer to any adversarial action that prevents the state estimator from receiving or using a particular measurement. Jamming may be conducted by several practical techniques including wireless jammers, GPS spoofers, coordinated Denial of Service attack \cite{ddos} or even by physical damage to the device, meter and communication equipment \cite{wsj}. Though jamming attacks have been implemented in research, there are few studies analyzing their impact on constructing optimal adversarial attacks. References \cite{infocom} and \cite{dekaPESGM2015} use jamming of flow measurements with attack on breaker statuses to design topology attacks on state estimation. The authors of a recent paper \cite{dekasmartgridcomm2015} have used measurement jamming to design `detectable' attacks. However, \cite{dekasmartgridcomm2015} limits adversarial action to insecure measurements and leaves encrypted measurements untouched. Though secure/encrypted measurements are indeed secure against data injection, they are jammable (Eg. though meter damage). Including jamming of secure measurements into the attack framework thus generalizes `detectable' and `hidden' attacks, and enables a complete analysis of the effect of jamming on attack cost and grid resilience. This is the principal focus of this work.

We develop \textit{\textbf{a graph-theoretic framework to study generalized `hidden' and `detectable' data attacks by an adversary equipped with three techniques. They include: (a) jamming and (b) data injection in insecure measurements, and (c) jamming of secure measurements.}} The distinct costs of these techniques will depend on the adversarial instruments and algorithms used for their implementation and measurement security available in the grid. Despite the possible variation in exact costs, we show that the design of the optimal attack depends only on the relative costs of jamming and injection. In particular, we show that
\begin{itemize}
\item for `hidden' attacks, the optimal generalized attack is given by the solution to a minimum weight graph-cut problem on a weighted graph, for all permissible costs of jamming and data injection;
\item for `detectable' attacks, the range of costs for the jamming and data injection tools can be divided into three intervals based on their relative values (Fig.~\ref{fig:regions}). In each cost region, the optimal generalized attack is constructed by solving at most two minimum weight constrained graph-cut problems specific to that interval.
\end{itemize}
It needs to be mentioned that if jamming is limited to insecure measurements, the optimal `detectable' attack is described by two cost intervals \cite{dekasmartgridcomm2015} with one graph-cut problem each, unlike three cost intervals, each with two optimization problems here. As the constrained graph cut problems are in general not solvable in polynomial time, we give iterative min-cut based approximate algorithms that can be used for attack construction. Simulations on IEEE test cases elucidate cost improvements produced by our generalized attack framework over traditional data attacks.

The second significant result of this paper states that \textit{\textbf{our generalized attacks are feasible even in systems with only one insecure measurement.}} Preventing them requires extending security to all measurements. Our attack framework is thus more potent than previously studied `hidden' \cite{deka,deka1} and `detectable' attacks \cite{dekaISGT,dekasmartgridcomm2015} that can be prevented with much less number of secure measurements as detailed later.

The rest of this paper is organized as follows. The next section presents a description of the system models used in state estimation, bad-data removal and considered adversarial tools and attack types. Traditional `hidden' and `detectable' attack regimes that involve manipulation of insecure measurements are discussed in Section~\ref{sec:attack}. Next, our generalized attack framework for `hidden' and `detectable' attacks is presented in Section~\ref{sec:generalized} along with graph-theoretic formulations to study the effect of different adversarial costs on the optimal attack design. The algorithms to design the optimal `hidden' and `detectable' generalized attacks are given in Section~\ref{sec:algo}. Simulations of the proposed algorithms on IEEE bus systems for a range of jamming and bad-data injection costs and comparisons with existing work are shown in Section~\ref{sec:results}. Finally, concluding remarks are presented in Section~\ref{sec:conclusion}.

\section{State Estimation, Bad-Data Removal \& Adversarial Action}
\label{sec:estimation}
The power grid represents a set $V$ of $n$ buses (nodes) connected by a set $E$ of $|E|$ transmission lines (directed edges). As an example, the IEEE $14$ bus test system \cite{testsystem} is given in Figure~\ref{fig:14bus}.
\begin{figure}
\centering
\includegraphics[width=0.5\textwidth, height = .31\textwidth]{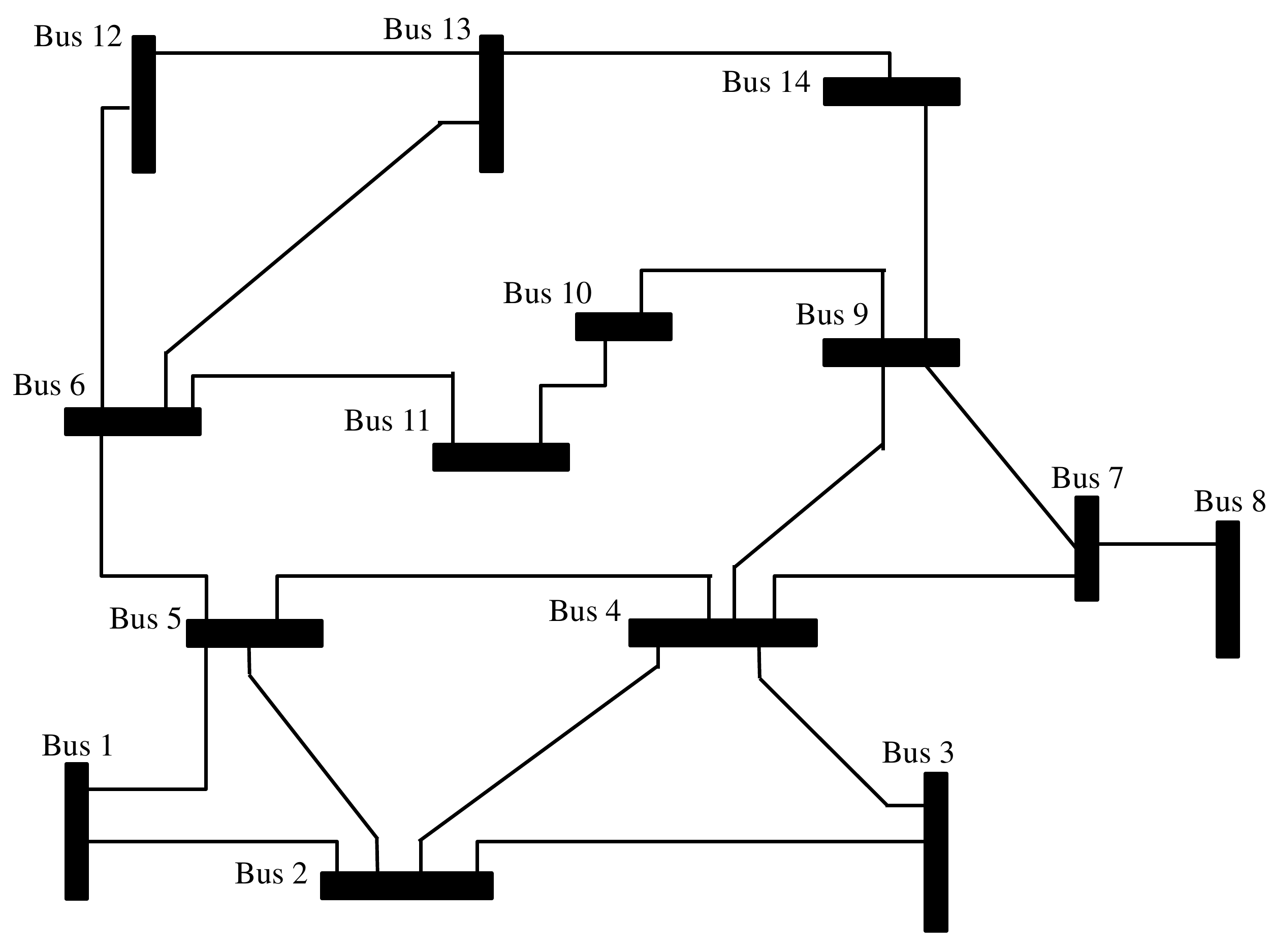}
\squeezeup
\squeezeup
\squeezeup
\caption{IEEE 14-bus test system \cite{testsystem}}
\label{fig:14bus}
\squeezeup
\end{figure}

\textbf{Measurement Model}: We use the DC power flow model \cite{bookDC} for the grid in this paper. Here, voltage magnitudes are assumed to be constant at unity on all buses and the state vector of the system comprises of all bus phase angles $x \in \mathbb{R}^{n}$. Transmission lines are assumed to be perfectly inductive (zero resistance) with a diagonal susceptance matrix $B$. We use $x_i$ to denote the phase angle at bus $i$ and $B_{ij}$ to denote the susceptance of line $(i,j)$. We consider a $m$ length measurement vector $z \in \mathbb{R}^m$  that comprises of  a) active power flows on lines and b) voltage phase angles on buses, collected from conventional meters and phasor measurement units in the grid. The relation between $z$ and $x$ is given by
\begin{align}
z = Hx + e \label{dcmodel}
\end{align}
where $H$ is the $m \times n$ full-ranked measurement matrix and $e$ is a zero mean Gaussian measurement noise vector with known covariance $\Sigma$. If the $k_1^{th}$ and $k_2^{th}$ entries (rows) in $z$ ($H$) measures the power flow on line $(i,j)$ and the phase angle at node $i$ respectively, then the DC power flow gives
\begin{align}
z(k_1) = B_{ij}(x(i)-x(j)),~ z(k_2) = x(i)\nonumber\\
H(k_1,:) = [0..0~~B_{ij}~~ 0..0~~-B_{ij}~~0..0] \label{line}\\
H(k_2,:) = [0..0~~1~~0..0~~0~~0..0] \label{bus}
\end{align}

Without a loss of generality, we introduce a ${n+1}^{th}$ `reference' bus with phase angle $0$ in the system and accordingly append $0$ to the state vector $x$. Note that the phase angle measurement at any bus $i$ is equivalent to a flow on a hypothetical line of unit conductance between bus $i$ and the `reference' bus. To represent this, we augment an additional column $h^{g}$ to matrix $H$ with value $-1$ for rows representing phase angles and $0$ otherwise. We thus convert every entry in $z$ into a flow measurement given by
\begin{align}
z = Hx = [H|h^g]\setlength{\arraycolsep}{2pt} \renewcommand{\arraystretch}{0.8}\begin{bmatrix} x \\0 \end{bmatrix} \nonumber
\end{align}
Note that the augmented measurement matrix has the structure of a susceptance weighted graph incidence matrix of rank $n$. From this point, we use $x$ and $H$ to denote the augmented state vector and measurement matrices respectively.\\
\textbf{State Estimation}: The complete DC state estimator used in this paper is given in Figure~\ref{estimator} \cite{monticelli,bookDC}.
\begin{figure}
\centering
\includegraphics[width=0.46\textwidth]{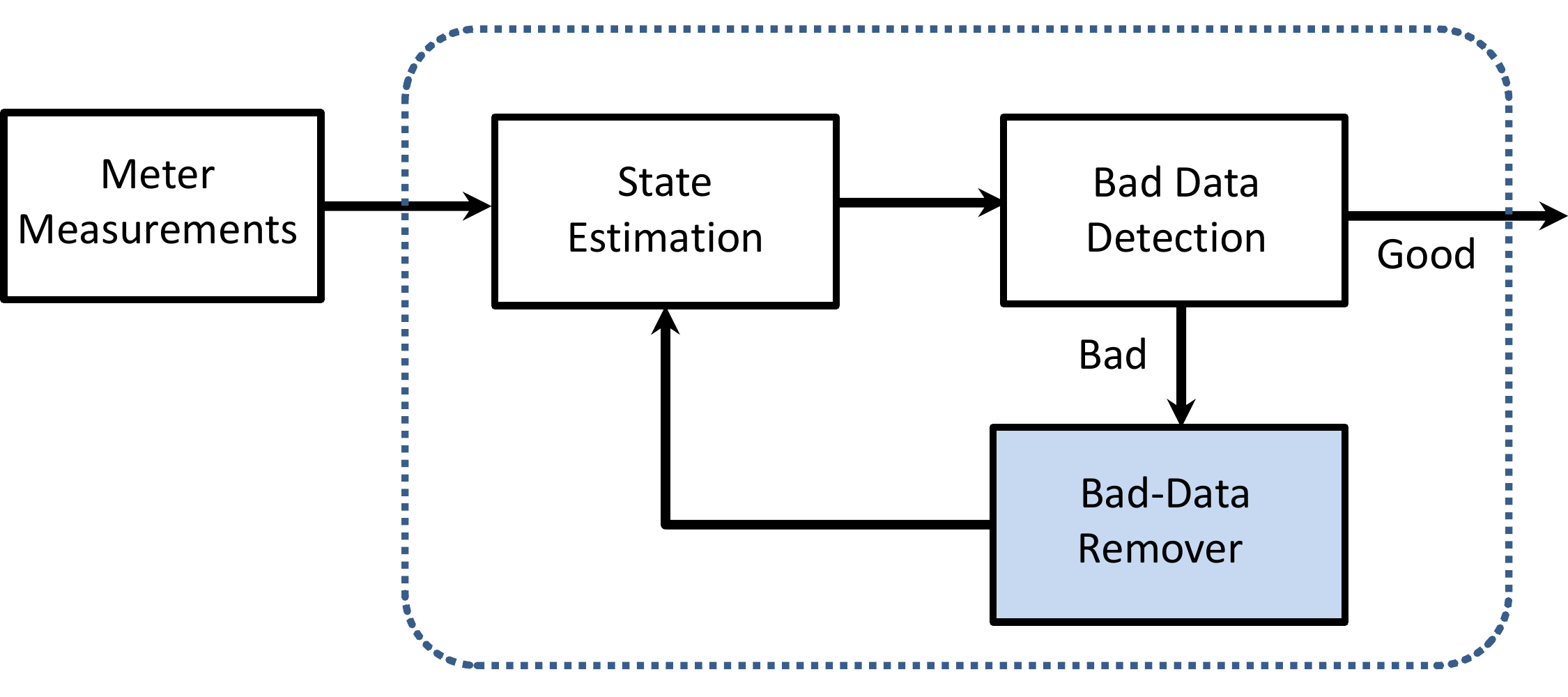}
\squeezeup
\caption{State Estimator for a power system \cite{monticelli,bookDC}}
\label{estimator}
\squeezeup
\end{figure}

The true state estimate $x^*$ is generated from measurement vector $z$ by a weighted least-square minimizer that minimizes the weighted residual's magnitude given by $J(x,z) = \|\Sigma ^{-.5}(z-Hx)\|_2$ over variable $x$. As shown in Fig.~\ref{estimator}, this step is followed by a threshold ($\lambda$) based bad-data detector that determines the presence of bad-data by the following test:
\begin{align}
\|\Sigma ^{-.5}(z-Hx^*)\|_2 &\leq \lambda ~~\text{accept~~~} x^*\nonumber\\
                            &> \lambda ~~\text{detect bad-data} \label{test}
\end{align}
If bad-data is detected, the bad-data remover is called to identify and remove bad-data as described below. \\
\textbf{Bad-data Removal:} Using basic linear algebra \cite{monticelli, bookDC}, it can be shown that the residual vector $r = z-Hx^* = [I - H(H^T\Sigma^{-1}H)^{-1}H^T\Sigma^{-1}]z$. Based on the assumption that probability of bad-data affecting greater number of locations is low, the estimator removes the minimum number of measurements such that the remaining measurements satisfy the bad-data check in Eq.~(\ref{test}). The optimal identification and removal scheme for multiple incorrect measurements is NP-hard \cite{monticelli, dekaISGT} and hence iterative or greedy schemes are used in practice. Unless otherwise stated, we assume that the unmanipulated measurement vector $z$ is clean and leads to estimation of the correct state vector $x^*$.\\

\textbf{Adversarial Tools and Attack Types:} Following past work in literature, we consider the adversary's goal to produce a non-zero change in the estimated state vector $x^*$ using an minimum cost attack. In reality, the adversary motivation may be economic (Eg. creating sub-optimal prices \cite{price}) or grid instability (Eg. producing/hiding grid failures) or be restricted to specific buses (Eg. targeted attacked \cite{deka1}). Keeping the adversarial goal as changing the state estimate analyzes the grid security in the strongest terms, where the grid controller is agnostic and gives equal weight to all adversaries.

We denote the secure set of measurements in $z$ that are encrypted against adversarial data injection by $S$. However, measurements in $S$ can suffer from bad-data arising from measurement noise. The remaining insecure measurements are denoted by set $S^c$. As stated in the Introduction, we consider three adversarial tools here. Among them, data injection is denoted by an additive vector $a$ that modifies the measurement vector $z$ to $z+a$. As secure measurements are immune to data injection, $a(i) = 0 ~\forall i \in S$. In contrast, jamming can be conducted on both secure and insecure measurements and is represented by removal of the jammed measurements from $z$. Let $p_I$, $p_J^{S^c}$ and $p_J^{S}$ denote the costs of data injection, jamming insecure measurements, and jamming secure measurements respectively. Further, a permissible set of costs are assumed to follow:\\
\textbf{Assumption $1$:} $p_J^{S^c} \leq p_J^{S} \leq p_I$\\
Note that data injection involves changing  meter measurements by precisely formatted real values and following communication protocols to ensure their usage at the state estimator. In constrast, jamming can be involved by physical \cite{wsj} or cyber destruction \cite{ddos} of the meter reading. Further, an adversary equipped with data injection can conduct jamming by inserting garbage values into the measurements. Thus, we assume that injection cost $p_I$ is not less that jamming costs. Secondly, jamming a secure measurement can be considered at least as costly as jamming an insecure measurement as secure measurements are encrypted and may require bypassing the resident security features leading to $p_J^{S^c} \leq p_J^{S}$. We assume the adversary to know/estimate these costs from the respective instrumentation and skills necessary for deployment. We show later that the attack construction depends on the relative values of these adversarial costs rather than their exact values.

A feasible attack refers to a successful attack; a feasible attack with minimum attack cost is called an optimal attack. We use \textbf{\textit{injection attacks}} to refer to attacks that use data injection alone. For attacks that additionally use jamming of insecure measurements, we use the phrase \textbf{\textit{jamming attacks}}. Attacks proposed in this work that use all three adversarial tools are termed \textbf{\textit{generalized attacks}}. Finally, we prefix the attack denotation by its `type'. The two types of attacks discussed in this paper are defined below.

\begin{definition} \label{attackdefine}
\item \textbf{`Hidden' attack} \cite{hidden}: This well-studied attack is not detected by the bad-data detector. The adversary ensures feasibility by manipulating measurements in a way such that the measurement residue remains unchanged.
\item \textbf{`Detectable' attack} \cite{frame,dekaISGT}: This attack initially fails the bad-data detection test but passes it after the estimator removes bad-data. The adversary ensures feasibility by manipulating measurements such that the minimum set of measurements that are removed to pass the detection test does not include all manipulated measurements.
\end{definition}

In the next section, we describe traditional attack frameworks (injection attacks and jamming attacks) that operate through insecure measurements only. This background will help analyze generalized data attacks in subsequent sections.

\section{Data Attacks using insecure measurements}
\label{sec:attack}
We analyze both `hidden' and `detectable' traditional (injection and jamming) attacks where the adversary is limited to attacking insecure measurements in set $S^c$. First, we focus on injection attacks.

\subsection{Injection Attacks}
\label{sec:injection}
Here, the adversary's strategy is entirely represented by the injection vector $a$ that is added to the measurement vector $z$. As data injection is the only tool available, its cost $p_I$ does not influence the attack construction. Consider the case of a `hidden' injection attack. As mentioned in Definition~\ref{attackdefine}, the attack is successful if it doesn't change the measurement residual. If $a = Hc \neq 0$ for some $c \in \mathbb{R}^{n+1}$, this holds  as $\|\Sigma ^{-.5}(z-Hx^*)\|_2 = \|\Sigma ^{-.5}(z+a-H(x^*+c))\|_2$ and the state estimate is modified to $x^*+c$. The optimal `hidden' injection attack is given by the sparest $a$ in the following \cite{deka,deka1}:
\begin{align} \label{hiddeninjection} \tag{H-I}
&\smashoperator[l]{\min_{c \in \mathbb{R}^{n+1}-\{\textbf{0}\}}} \|a\|_{0} \\
\text{s.t. ~} &a =Hc, c(n+1) = 0, ~a(i) = 0 ~\forall i \in S ~~(\text{$S$: Secure Set}) \nonumber
\end{align}
Next, we look at a `detectable' injection attack. By Definition~\ref{attackdefine} and the state estimator's bad-data removal technique described after Eq.~\ref{test}, it is clear that an injection vector $a \neq 0$ will successfully change the state estimate only if removal of some $k < \|a\|_0$ entries from the measurement vector is sufficient to pass the bad-data detection test, while preserving observability. We describe the construction of such an injection vector $a$ now. For any $Hc \neq 0$, include \textit{more than half} of the non-zero entries in $Hc$ in $a$ and replace the rest by $0$. Observe that $\|a\|_0 > \|Hc-a\|_0$ here. Thus, measurements corresponding to the non-zero terms in $(Hc-a)$ are incorrectly identified as bad-data instead of the injected measurements in $a$. After removal of bad-data from measurement vector and elimination of corresponding rows from the measurement matrix $H$, $a$ now lies in the column space of the modified measurement matrix and a feasible attack is conducted. The optimal measurements from $Hc$ to include in the attack vector $a$ are given by the unity terms in the optimal binary vector $d$ of the following \cite{dekaISGT,dekasmartgridcomm2015}:
\begin{align} \label{detectableinjection} \tag{D-I}
&\smashoperator[l]{\min_{d \in \{0,1\}^m, c \in \mathbb{R}^{n+1}- \{\textbf{0}\}}} \|d\|_{0} \\
\text{s.t. ~} &c(n+1) = 0, d(i) = 0 ~\forall i \in S \nonumber\\
& \|d\|_{0} > \|Hc\|_0/2 ~~(\text{for feasibility})\label{cond1} \\
& rank(DH) = n,~ diag(D) = \textbf{1} - (\textbf{1}-d)*(Hc)_{spty} \label{cond2}
\end{align}
Here, $a*b$ refers to the element-wise multiplication between vector $a$ and $b$, while $a_{spty}$ denotes the sparsity pattern in vector $a$. In the rank constraint (\ref{cond2}), $D$ is a $0-1$ diagonal matrix with value of $0$ for removed measurements. We now describe graph-theoretic solutions for attack construction for both attack types.

\textbf{Graph-Theoretic Solution}: We construct undirected graph $G_H$ with $n+1$ nodes and edges corresponding to measurement rows in $H$. We denote secure and insecure edges in $G_H$ corresponding to secure and insecure measurements in $H$ respectively. Due to the unimodular structure of $H$, it can be shown that the optimal solutions of Problems~\ref{hiddeninjection} or~\ref{detectableinjection} remain unchanged if $c$ is restricted to be a $0-1$ binary vector and $H$ is replaced by the unweighted incidence matrix $A_H$ of $G_H$. In this case, the non-zero terms in $A_Hc$ in fact correspond to a graph-cut in $G_H$ \cite{deka,dekaISGT}. Thus, the optimal attack design can be stated as a graph cut problem as described below:

\begin{theorem}\label{injectiondesign}
\item {\cite[Theorem 2]{deka}} The optimal `hidden' injection attack in Problem~\ref{hiddeninjection} is given by the minimum cardinality cut $C^*$ in $G_H$ with no secure edges.
\item {\cite[Theorem 2]{dekaISGT}} The optimal `detectable' injection attack in Problem~\ref{detectableinjection} is given by any $\lfloor1+ |C^*|/2\rfloor$ insecure edges in $C^*$, where $C^*$ denotes the minimum cardinality cut in $G_H$ with a minority of secure edges ($|C^* \cap S| < |C^*|/2 $).
\end{theorem}
It follows immediately that the cost of the optimal `detectable' injection attack is never greater than $.5+ 1/|C^*|$ times the cost of the optimal `hidden' injection attack $C^*$. Next, we add jamming of insecure measurements to the attack framework and discuss its implications.

\subsection{Jamming Attacks}
\label{sec:jamming}
Here the adversary can jam and remove insecure measurements at a cost $p_J^{S^c}$ in addition to injecting data at cost $p_I$. Note that for a non-zero change in state estimate, adversary should inject data into at least one insecure measurement. The design of the optimal `hidden' jamming attack is given by:
\begin{theorem}\label{hiddenjammingdesign}
The optimal `hidden' jamming attack for all permissible $p_I$ and $p_J^{S^c}$ is constructed by injecting data into one edge and jamming the remaining edges in the minimum cardinality cut in $G_H$ with no secure edges.
\end{theorem}
\textbf{Brief Proof steps}: Using Theorem~\ref{injectiondesign} and $p_J^{S^c} \leq p_I$, it is clear that the least cost `hidden' jamming attack designed using the optimal `hidden' injection attack is given by Theorem~\ref{hiddenjammingdesign}. Its global optimality can be proved by contradiction.

Now we look at `detectable' jamming attacks as discussed in \cite{dekasmartgridcomm2015}. Consider a cut $C$ in graph $G_H$ with $n^S_C$ and $n^{S^c}_C$ secure and insecure edges respectively, with $n^{S^c}_C > n^S_C$. Using Theorem~\ref{injectiondesign}, $C$ is feasible for a `detectable' injection attack. If the adversary jams $k^C_J < n^{S^C}_C- n^S_C$ insecure measurements in $C$, the remaining $|C| - k^C_J$ measurements still constitute a feasible cut with a majority of insecure edges. The adversary can thus inject data into $\lfloor 1 + \frac{|C|-k^C_J}{2}\rfloor$ insecure edges in $C$ to conduct a successful `detectable' jamming attack of attack cost $p^C$ given by
 \begin{align}
 p^C &= p_J^{S^c}k^C_J+p_I\lfloor1 + \frac{|C|-k^C_J}{2}\rfloor\nonumber\\
 &= (p_J^{S^c}-p_I/2)k^C_J+p_I\frac{|C|+2-(|C|- k^C_J)\mod 2}{2}\label{cost1a}
 \end{align}
Note that if $p_J^{S^c} < p_I/2$ the attack cost $p^C$ of cut $C$ decreases with increasing $k^C_J$ and is lowest at $k^C_J= n^{S^c}_C-n^S_C-1$. Similarly, it can be shown that for $p_J^{S^c} \geq p_I/2$, the attack cost is minimum at $k^C_J=1 -|C|\mod 2$. Using these values of $k^C_J$ in the Eq.~(\ref{cost1a}) for attack cost leads to the following result on optimal attack design.
\begin{theorem}\label{detectablejammmingdesign}\cite[Theorem 2]{dekasmartgridcomm2015}
The construction of the optimal `detectable' jamming attack  with jamming cost $p_J^{S^c}$ and data injection cost $p_I$ for insecure measurements is given by:
\item $1$. \textbf{$p_J^{S^c} < p_I/2$:} Give weights of $p_I-p_J^{S^c}$ and $p_J^{S^c}$ to secure and insecure edges respectively in $G_H$ and find the minimum weight cut $C^*$ with $n_{C^*}^{S} < |C^*|/2$ secure edges. Use $1+n_{C^*}^{S}$ insecure edges for bad-data injection and jam the other insecure edges.
\item $2$. \textbf{$p_J^{S^c} \geq p_I/2$:} Find the minimum cardinality cut $C^*$ with a minority of secure edges in $G_H$. Use $\lfloor \frac{1+|C^*|}{2} \rfloor$ insecure edges for data injection and jam $(1 - |C^*|\mod 2)$ insecure edges.
\end{theorem}
A detailed derivation of Theorem~\ref{detectablejammmingdesign} is given in \cite{dekasmartgridcomm2015}. The main argument is also elucidated through the example in Fig.~\ref{fig:feasibleattack}. To conclude, the range of permissible relative costs for jamming insecure measurements is thus separable into two intervals with distinct designs for optimal `detectable' jamming attack.

\begin{figure}[ht]
\centering
\includegraphics[width=0.50\textwidth]{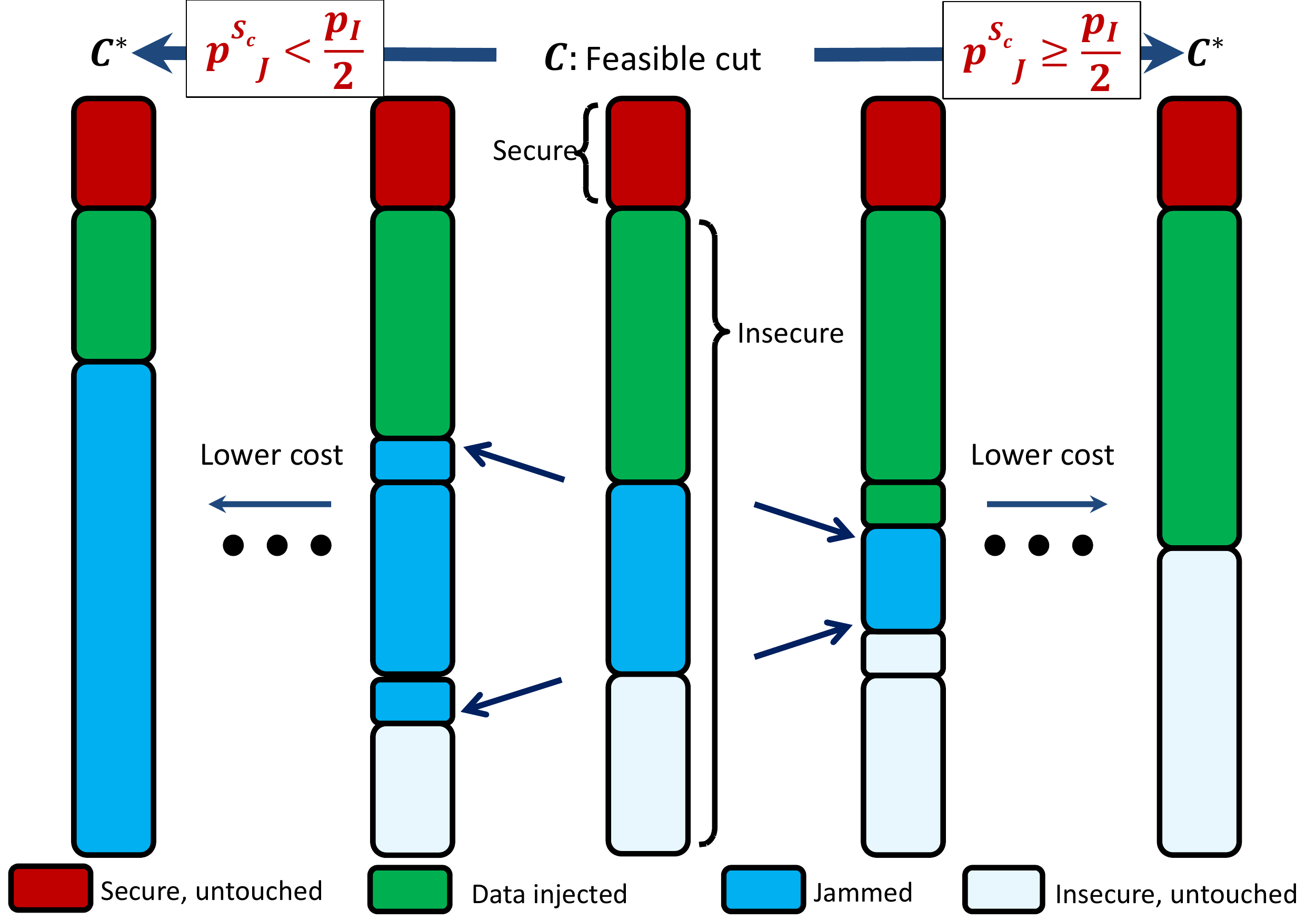}
\caption{Effect of jamming cost $p_J^{S^c}$ and data injection cost $p_I$ on the minimum `detectable' jamming attack $C^*$ derived from a cut $C$ with $n^S_C (< |C|/2)$ secure and $n^{S^c}_C$ insecure edges. Secure edges in the cut are colored red while untouched, jammed and injected insecure edges are colored white, blue and green colors respectively. When $p_J^{S^c}<p_I/2$, attack cost can be reduced by replacing one data injection with jamming two measurements as shown in the cuts on the left of $C$. For $p_J^{S^c} \geq p_I/2$, attack cost is reduced by replacing two jammed measurements by one measurement with data injection while leaving the other untouched as shown on the right side of cut $C$. Optimal attacks $C^*$ got from this replacement are given by Theorem~\ref{detectablejammmingdesign}.}
\label{fig:feasibleattack}
\end{figure}

In the next section, we present our generalized attack framework that allows jamming (not data injection) of secure measurements by the adversary.

\section{Data Attacks with Jamming secure measurements}
\label{sec:generalized}
The adversary in this case has three tools (jamming secure measurement, jamming insecure measurement, and data injection in insecure measurement) with distinct costs per measurement ($p_J^{S}, p_J^{S^c},$ and $p_I$). From Assumption $1$, we have $p_J^{S^c} \leq p_J^{S}\leq p_I$. The introduction of jamming of secure measurements creates major changes in the adversarial strategy as it relaxes the feasibility requirements for both `hidden' and `detectable' attacks as noted below.

\subsection{`Hidden' Generalized Attacks}
\label{sec:hiddengeneralized}
Theorem~\ref{injectiondesign} and~\ref{hiddenjammingdesign} states that feasible cuts for `hidden' injection and jamming attacks cannot include any secure edge. With the ability to jam secure measurements, this is no longer necessary. Consider a cut $C$ with $n^S_C$ secure and $n^{S^c}_C >0$ insecure edges. If all $n^{S}_C$ secure edges are removed by jamming, the remaining cut can provide a `hidden' attack where one insecure edge is used for data injection and the rest are jammed. The total attack cost is $p_J^{S}n^S_C+p_J^{S^c}n^{S^c}_C + (p_I-p_J^{S^c}) $. The optimal attack is thus given by:

\begin{theorem}\label{hiddengeneralizeddesign}
Give weights of $p_J^{S}$ and $p_J^{S^c}$ to secure and insecure edges respectively in $G_H$ and find the minimum weight cut $C^*$ with non-zero number of insecure edges. The optimal `hidden' generalized attack is constructed by using one insecure edge in $C^*$ for data injection and jamming the remaining cut-edges.
\end{theorem}
Note that the optimal attack design here has the same form for all relative values of jamming and injections costs. Next, we look at `detectable' generalized attacks.

\subsection{`Detectable' Generalized Attacks}
\label{sec:detectablegeneralized}
We study how the design of a `detectable' attack changes when jamming of secure measurements is allowed. To do so, we consider a cut $C$ in graph $G_H$ with $n^S_C$ secure and $n^{S^c}_C$ insecure edges. We can have two cases for $C$: A) $n^S_C <n^{S^c}_C$ and B) $n^S_C \geq n^{S^c}_C$. Theorem~\ref{injectiondesign} and~\ref{detectablejammmingdesign} state that to conduct a successful `detectable' injection or jamming attack, the adversary requires graph-cuts with a majority of insecure edges (Case A). Thus, we have
\begin{lemma}\label{lemma1}
A `detectable' generalized attack can be constructed from cut $C$ having $n^S_C$ secure and $n^{S^c}_C$ insecure edges if $n^{S^c}_C >0$ and the adversary initially jams $k^S_C \geq [n^S_C - n^{S^c}_C + 1]^{+}$ secure cut-edges, where $[a]^+ = \max\{0,a\}$.
\end{lemma}
This step ensures that after removal of $k^S_C$ jammed secure measurements, the remaining cut has a majority of insecure edges as shown in Fig.~\ref{fig:jammingsecure}. Further, jamming of secure edges can lead to a reducing in attack cost as well. For example, if $p_J^{S^c}+p_J^S \leq p_I$, a feasible cut $C$'s data injected insecure edge can be replaced with jamming of two edges in $C$, one secure and another insecure to lower the attack cost. This is demonstrated by the cut on the right side in Fig.~\ref{fig:jammingsecure}. 
\begin{figure}[ht]
\centering
\includegraphics[width=0.50\textwidth]{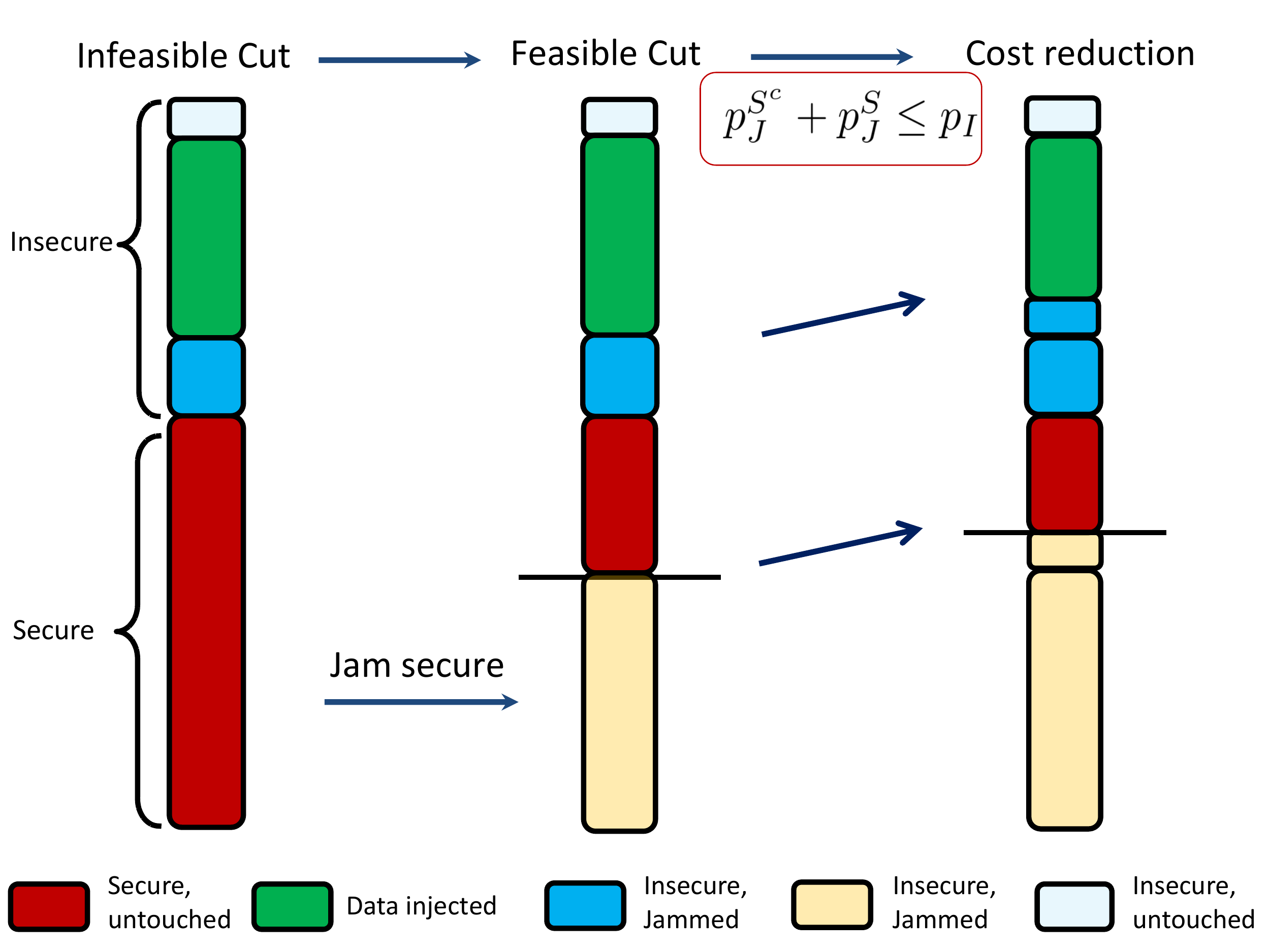}
\caption{Effect of jamming secure measurements on design of `detectable' generalized attacks. The cut on the left is infeasible due to a minority of insecure edges. Jamming secure measurements leads to a feasible cut in the middle. Finally for $p_J^{S^c}+p_J^S \leq p_I$, attack cost can be reduced by replacing one data injected edge with two jammed edges (one secure and one insecure).}
\label{fig:jammingsecure}
\end{figure}

To analyze the effect of jamming cost for secure edges, we follow the approach in  Theorem~\ref{detectablejammmingdesign}. We aim to determine the optimal `detectable' generalized attack strategy over different range of costs for $p_I,p_J^S$ and $p_J^{S^c}$. We begin with the following cost interval.\\ \\
\textbf{Cost Interval I}:$[p_J^{S_c}\geq p_I/2] \bigcap [p_J^{S} \geq p_I/2]$\\ Using Theorem~\ref{detectablejammmingdesign} for $p_J^{S_c}\geq p_I/2$, the minimum cost attack using the remaining $|C|-k^S_C$ edges is constructed by injecting data into $\lfloor\frac{1 + |C|- k^{S}_C}{2}\rfloor$  and jamming $(1 - (|C|- k^{S}_C)\mod 2)$ insecure edges. The total cost is given by:
\begin{align}
p^C = p_J^{S}k^S_C + p_I\lfloor\frac{1 + |C|- k^{S}_C}{2}\rfloor + p_J^{S^c}(1 - (|C|- k^{S}_C)\mod2) \label{cost}
\end{align}
As $p_J^{S} \geq p_I/2$, we note that $p^C $ is increasing with $k^S_C$. Using Lemma~\ref{lemma1}, the minimum cost is achieved at $k^S_C = [n^S_C - n^{S^c}_C + 1]^{+}$. For Case A ($n^S_C <n^{S^c}_C$), this gives $k^S_C = 0$ (no jamming of secure measurement), and the optimal attack is identical in structure to the optimal `detectable' jamming attack for $[p_J^{S_c} \geq p_I/2]$ in Theorem~\ref{detectablejammmingdesign}. For Case B ($n^S_C \geq n^{S^c}_C$), the optimal $k^S_C$ equals $n^S_C - n^{S^c}_C + 1$. The attack cost thus reduces to
\begin{align}
p^C &= p_J^{S}(n^S_C - n^{S^c}_C + 1) + p_In^{S^c}_C ~~(\text{using Eq.~(\ref{cost})})\nonumber\\
    &= p_J^{S}n^S_C + (p_I -p_J^{S})n^{S^c}_C + p_J^{S}\label{cost1}
\end{align}
Excluding the constant term, this optimal attack cost for $C$ in Case B is equal to its cut-weight if secure and insecure edges are given weights $p_J^{S}$ and $p_I - p_J^{S}$ respectively.

As $G_H$ has cuts in both Case A and Case B, the optimal `generalized' attack selects the minimum cost one among the optimal attacks for Cases A and B. This is summarized below:
\begin{theorem}\label{detectablegeneralizeddesignI}
The optimal `detectable' generalized attack in $G_H$ for the cost interval $[p_J^{S_c}\geq p_I/2] \bigcap [p_J^{S} \geq p_I/2]$ is given by the minimum cost attack among the optimal solutions of the following two graph optimization problems:
\item \emph{Problem I-A}. Find the minimum cardinality feasible cut $C^*$ in $G_H$ with a minority of secure edges. Use $\lfloor (1+|C^*|)/2 \rfloor$ insecure edges for bad-data injection and jam $(1 - |C^*|\mod 2)$ insecure edges.
\item \emph{Problem I-B}. Give weights of $p_J^{S}$ and $p_I - p_J^{S}$ to secure and insecure edges respectively in $G_H$ and find the minimum weight cut $C^*$ with $(n_{C^*}^{S} \geq |C^*|/2)$ secure edges and $(n_{C^*}^{S^c} > 0)$ insecure edges. Inject data into all insecure edges and jam $(n_{C^*}^{S}+1- n_{C^*}^{S^c})$ secure edges.
\end{theorem}

Next we analyze cut $C$ with $n^S_C$ secure and $n^{S^c}_C >0$ insecure edges in the second cost interval.\\\\
\textbf{Cost Interval II}: $[p_J^{S_c} < p_I/2] \bigcap [p_J^{S}+ p_J^{S_c} \geq p_I]$\\
By Lemma~\ref{lemma1}, the adversary initially jams $k^S_C \geq [n^S_C - n^{S^c}_C + 1]^{+}$ secure cut-edges leaving $(n^S_C - k^S_C)$ secure and $n_{C^*}^{S^c}$ insecure edges.  As $p_J^{S_c} < p_I/2$, the minimum cost attack constructed from the remaining edges includes data injection into $n^S_C - k^S_C+1$ measurements and jamming the rest of the insecure measurements (see Theorem~\ref{detectablejammmingdesign}). This gives an attack cost of:
\begin{align}
p^C &= p_J^{S}k^S_C + p_I(n^S_C - k^S_C+1) + p_J^{S^c}(n^{S^c}_C - n^S_C + k^S_C-1) \nonumber\\
& = (p_J^{S}+ p_J^{S^c}-p_I)k^S_C + (p_I-p_J^{S^c})(n^S_C+1) + p_J^{S^c}n^{S^c}_C\label{cost2}
\end{align}
As $p_J^{S}+ p_J^{S_c} \geq p_I$, the attack cost in Eq.~(\ref{cost2}) increases with $k^S_C$. The minimum attack cost is thus attained for Case A ($n^S_C <n^{S^c}_C$) at $k^S_C = 0$, and for Case B ($n^S_C \geq n^{S^c}_C$) at  $k^S_C =  n^S_C - n^{S^c}_C + 1$. The corresponding attack costs are given by:
\begin{align}
p^C &= (p_I-p_J^{S^c})(n^S_C+1) + p_J^{S^c}n^{S^c}_C ~~(\text{for Case A})\label{cost3}\\
p^C &= p_J^{S}(n^S_C+1) + (p_I - p_J^{S})n^{S^c}_C   ~~(\text{for Case B})\label{cost4}
\end{align}
Observe that in either case, ignoring additive constants, the optimal attack cost is given by the cut-weight of $C$ with distinct weights for secure and insecure measurements. We can thus determine the optimal `detectable' generalized attack in this interval as follows:

\begin{theorem}\label{detectablegeneralizeddesignII}
The optimal `detectable' generalized attack in $G_H$ for the cost interval $[p_J^{S_c} < p_I/2] \bigcap [p_J^{S}+ p_J^{S_c} \geq p_I]$ is given by the minimum cost attack among the optimal solutions of the following two graph optimization problems:
\item \emph{Problem II-A}. Give weights of $p_I - p_J^{S^c}$ and $p_J^{S^c}$ to secure and insecure edges respectively in $G_H$ and find the minimum weight cut $C^*$ with ($n_{C^*}^{S} < |C^*|/2$) secure edges. Inject data into $(n_{C^*}^{S}+1)$ insecure edges and jam the other insecure edges.
\item \emph{Problem II-B}. Give weights of $p_J^{S}$ and $p_I - p_J^{S}$ to secure and insecure edges respectively in $G_H$ and find the minimum weight cut $C^*$ with $(n_{C^*}^{S} \geq |C^*|/2)$ secure edges and $(n_{C^*}^{S^c} > 0)$ insecure edges. Inject data into all insecure edges and jam $(n_{C^*}^{S}+1- n_{C^*}^{S^c})$ secure edges.
\end{theorem}

Finally, we look at cost interval III with low jamming costs.\\\\
\textbf{Cost Interval III}:$[p_J^{S_c} < p_I/2] \bigcap [p_J^{S}+ p_J^{S_c} < p_I]$\\
As $p_J^{S_c} < p_I/2$ constraint is common to Interval II, the preliminary analysis here is identical to the discussion preceding Eq.~(\ref{cost2}) and leads to the following attack cost:
\begin{align}
p^C &= (p_J^{S}+ p_J^{S^c}-p_I)k^S_C + (p_I-p_J^{S^c})(n^S_C+1) + p_J^{S^c}n^{S^c}_C\label{cost5}
\end{align}
where $k^S_C \geq [n^S_C - n^{S^c}_C + 1]^{+}$ is the number of jammed secured measurements. Observe that the attack cost \emph{decreases} on increasing $k^S_C$ in this Interval. The minimum attack cost is thus obtained when $k^S_C = \max k^S_C = n^S_C$ for both Cases A and B. The optimal attack cost for cut $C$ is given by:
\begin{align}
p^C &= p_J^{S}n^S_C + p_J^{S^c}n^{S^c}_C+ (p_I-p_J^{S^c})~~(\text{for Cases A, B})\label{cost6}
\end{align}
which is an additive constant away from $C$' cut-weight if secure and insecure edges are given weights $p_J^S$ and $p_J^{S^c}$ respectively. The optimal `detectable' generalized attack design is given by the following theorem.
\begin{theorem}\label{detectablegeneralizeddesignIII}
The optimal `detectable' generalized attack in $G_H$ for the cost interval $[p_J^{S_c} < p_I/2] \bigcap [p_J^{S}+ p_J^{S_c} < p_I]$ is given by the optimal solution of the following graph optimization problem:
\item \emph{Problem III}. Give weights of $p_J^{S}$ and $p_J^{S^c}$ to secure and insecure edges respectively in $G_H$ and find the minimum weight cut $C^*$ with non-zero insecure edges. Inject data into one insecure edge and jam all other secure and insecure edges.
\end{theorem}

\begin{figure}[ht]
\centering
\includegraphics[width=0.5\textwidth]{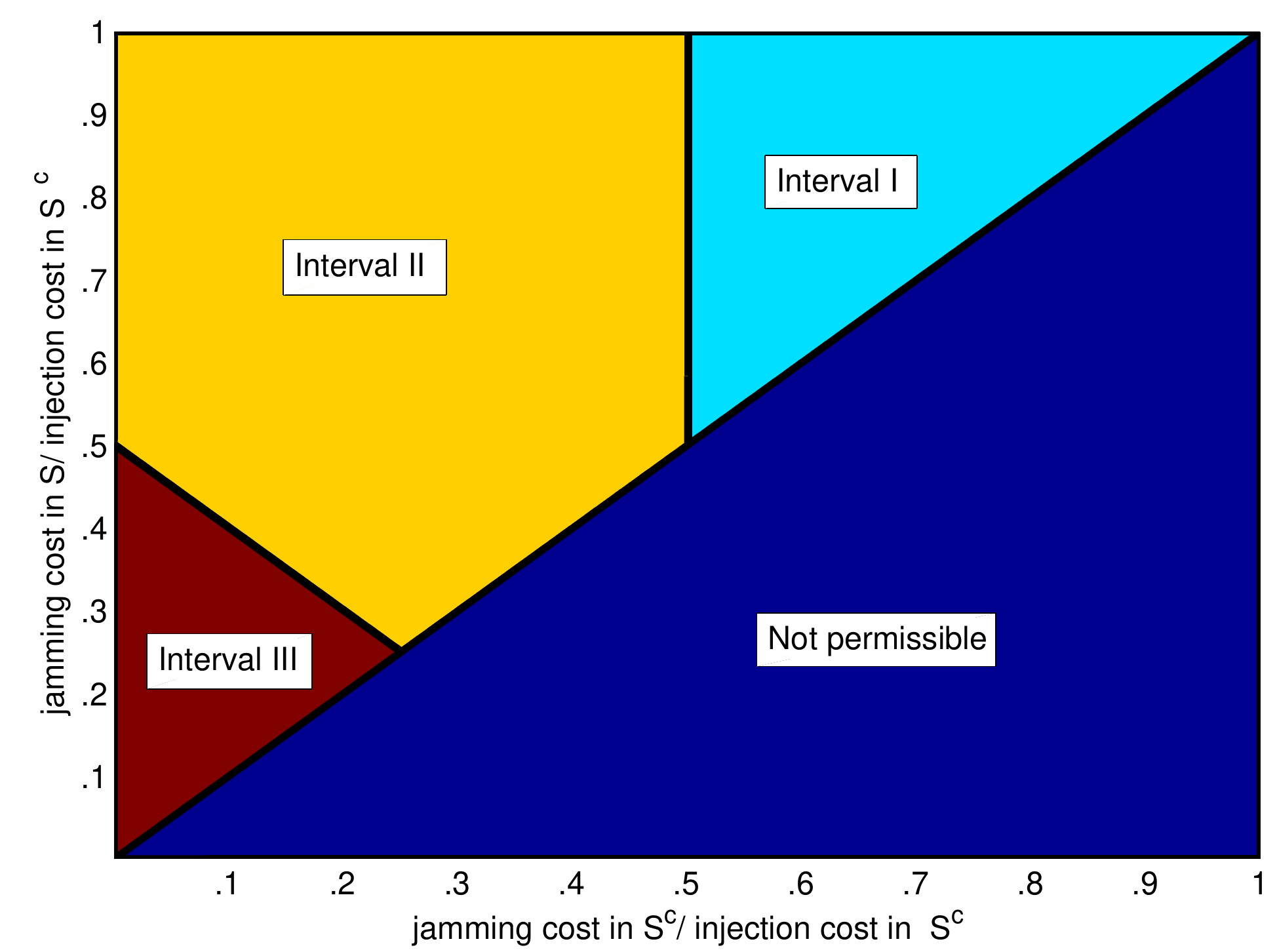}
\caption{Separation of the range of relative costs for jamming secure ($p_J^{S}$) and insecure ($p_J^{S_c}$) measurements into intervals with distinct formulations for optimal `detectable' generalized attack. Interval I denotes $[p_J^{S_c}\geq p_I/2] \bigcap [p_J^{S} \geq p_I/2]$, Interval II denotes $[p_J^{S_c} < p_I/2] \bigcap [p_J^{S}+ p_J^{S_c} \geq p_I]$ and Interval III denotes $[p_J^{S_c} < p_I/2] \bigcap [p_J^{S}+ p_J^{S_c} < p_I]$. The fourth interval $p_J^{S} < p_J^{S_c}$ is not permissible by Assumption $1$.}
\label{fig:regions}
\end{figure}

To summarize, the design of the optimal `detectable' generalized attack can be divided into three intervals that cover the entire range of permissible jamming and data injection costs as shown in Fig.~\ref{fig:regions}. In Internals I (Theorem~\ref{detectablegeneralizeddesignI}) and II (Theorem~\ref{detectablegeneralizeddesignII}), the optimal attack is given by the minimum of two constrained graph-cut problems, while in Interval III (Theorem~\ref{detectablegeneralizeddesignIII}), it is given by the solution of a single problem. The following points are worth noting.
\begin{enumerate}
\item  Problems I-A and II-A pertaining to Case A in Intervals I and II are identical to the sub-problems for designing optimal `detectable' jamming attacks in Theorem~\ref{detectablejammmingdesign}.
\item Problems I-B and II-B pertaining to Case B in Intervals I and II are identical.
\item Problem III in Interval III is identical to the problem of designing optimal `hidden' generalized attacks in Theorem~\ref{hiddengeneralizeddesign}.
\end{enumerate}
The first two observations arise from the constraint $p_J^{S}+ p_J^{S_c} \geq p_I$ in Intervals I and II. This constraint restricts the optimal number of jammed secured measurements at the minimum necessary for feasible attack construction, which is $0$ for cuts with majority of insecure edges. Thus Problems I-A and II-A are similar to the ones in Theorem~\ref{detectablejammmingdesign}. For Interval III, the constraint $p_J^{S}+ p_J^{S_c} < p_I$ implies that the attack cost can be reduced by replacing data injection at one measurement with jamming of a pair of insecure and secure measurements or jamming two insecure measurements. Thus, the optimal `detectable' generalized attack in Interval III includes only one measurement with data injection and is identical to the optimal `hidden' generalized attack in Theorem~\ref{hiddengeneralizeddesign}. 

For all permissible costs as per Assumption $1$, the reduction in attack cost as a result of jamming is shown through simulations in Section~\ref{sec:results}. In addition, the next theorem presents the threat to grid resilience posed by generalized attacks.

\begin{theorem}\label{protection}
\item $1$. A system is vulnerable to generalized data attacks (both `hidden' and `detectable') even if it contains only one insecure measurement.
\item $2$. Addition of new secure measurements alone does not prevent generalized attacks.
\end{theorem}
\begin{proof}
Consider the graph $G_H$. As mentioned in Theorems~\ref{hiddengeneralizeddesign},~\ref{detectablegeneralizeddesignI},~\ref{detectablegeneralizeddesignII} and~\ref{detectablegeneralizeddesignIII}, a feasible generalized attack requires a cut in $G_H$ with non-zero number of insecure edges. Such a cut does not exist only if all measurements are secure. Hence the first statement holds. Addition of new secure measurements can increase the attack cost of a cut but does not change its feasibility. Hence the second statement holds.
\end{proof}
It follows from Theorem~\ref{protection} that the prevention of generalized attacks needs all existing insecure measurements to be \emph{replaced} with secure ones, rather than \emph{addition} of new secure measurements. This is a much stricter requirement than that for traditional `hidden' and `detectable' attacks which can be prevented by adding $n$ and $O(m/2)$ new secure measurements respectively \cite{deka1,dekaISGT}. Here, $n$ is the number of buses (excluding `reference' bus) and $m$ is number of measurements in the grid. Thus, our generalized attack framework undermines grid resilience to data attacks and cyber adversaries beyond previously studied attack models. In the next section, we comment on the hardness of designing generalized data attacks and develop approximate iterative algorithms to solve them.

\section{Algorithm For Generalized Attack Construction}
\label{sec:algo}
Consider the graph $G_H$ with sets $S$ and $S^c$ of secure and insecure edges respectively. The adversary is assumed to know the costs associated with jamming an insecure measurement, jamming a secure measurement and injecting data into an insecure measurement, given by $p_J^{S^c}, p_J^{S}$ and $p_I$ respectively. We first discuss algorithm for designing `hidden' generalized attacks.

\textbf{`Hidden' generalized attacks}: By Theorem~\ref{hiddengeneralizeddesign}, the optimal attack of this type is given by the minimum weight cut $C^*$ with non-zero insecure edges in $G_H$, where secure and insecure edges have weight $p_J^S$ and $p_J^{S^c}$ respectively. Algorithm $1$ outputs the optimal attack.

\begin{algorithm}
\caption{Optimal `Hidden' Generalized Attack Design}
\textbf{Input:} Graph $G_H$, Set $S$ ($S^c$) of secure (insecure) edges with edge-weights $p_J^{S}$ ($p_J^{S^c}$)
\begin{algorithmic}[1]
\STATE $i \gets 1, w \gets \infty$  \label{algo1temp}
\WHILE{$i \leq  |S^c|$}\label{algo1cond1}
\STATE Pick $i^{th}$ edge $(s,t)$ in $S^c$.
\STATE $C \gets$ minimum weight `$s-t$' cut separating  $s$ and $t$ in $G_H$
\IF{$w > C$'s weight}
\STATE $w \gets C$'s weight, $C_f \gets C$
\ENDIF
\STATE $i \gets i+1$
\ENDWHILE
\STATE Use $C_f$ for optimal attack in Theorem~\ref{hiddengeneralizeddesign}.
\end{algorithmic}
\end{algorithm}

\textbf{Working and Complexity}: In each iteration of the While Loop (Step~\ref{algo1cond1}), Algorithm $1$ picks an insecure edge in $S^c$ and finds the minimum weight cut $C$ that contains it. The feasible cut $C_f$ is updated if the current cut $C$ has lower weight. At the end of the iteration, the optimal attack is constructed by injecting data into one insecure edge and jamming the rest of the edges in $C_f$. Since, minimum `$s-t$' cut can be computed using max-flow algorithm in $O(nm\log(n^2/m))$ time \cite{maxflow}, Algorithm $1$ has polynomial time complexity of $O(|S^c|nm\log(n^2/m))$. Here $n$ and $m$ are number of nodes and edges in graph $G_H$.


\textbf{`Detectable' generalized attacks}: As analyzed in the previous section, the relative values of costs of jamming and data-injection change the design of `detectable' generalized attacks. Attack construction in Interval III is identical to that of `hidden' generalized attacks and is solved in polynomial time by Algorithm $1$. Here, we discuss the construction of attacks in Intervals I ($[p_J^{S_c}\geq p_I/2] \bigcap [p_J^{S} \geq p_I/2]$) and II ($[p_J^{S_c} < p_I/2] \bigcap [p_J^{S}+ p_J^{S_c} \geq p_I]$). Theorems~\ref{detectablegeneralizeddesignI} and~\ref{detectablegeneralizeddesignII} state that in either interval, the optimal `detectable' generalized attack is determined by solving two constrained graph-cut problems on $G_H$. In each of these problems (I-A, I-B, II-A and II-B), the constraint involves finding a cut $C$ in $G_H$ of Case A($n^{S}_C < |C|/2$) or Case B($n^{S}_C \geq |C|/2$) where $n^S_C$ is the number of secure edges in the cut. Reference \cite{dekaISGT} states that finding a cut where edges of one kind are in majority is equivalent to the NP-hard `ratio-cut' problem \cite{ratio}. Thus, determining the optimal `detectable' generalized attack in Intervals I and II is NP-hard in general.

Now, we provide an approximate algorithm (Algorithm $2$) for solving constrained graph-cut problems of the form included in Theorems~\ref{detectablegeneralizeddesignI} and~\ref{detectablegeneralizeddesignII}. Algorithm $2$ is a generalization of an iterative min-cut based algorithm in \cite{dekasmartgridcomm2015}, with additional constraints. The exact weights for secure and insecure edges and constraint (Case A or B) are specified by the particular problem being solved.

\begin{algorithm}
\caption{`Minimum Weight Constrained Graph-Cut Construction}
\textbf{Input:} Graph $G_H$, Set $S$ and $S^c$ of secure and insecure edges respectively, edge weights and Case (A or B) given by problem (I-A, I-B, II-A or II-B), $\beta,\gamma$ \\
\begin{algorithmic}[1]
\STATE Compute min-weight cut $C$ in $G_H$ \label{step1}
\STATE $w_C \gets C$'s weight
\IF{Case A}
\WHILE {($w_C < \gamma )\&\& (2|C \bigcap S| \geq |C|$)} \label{step2a}
\STATE Pick  $i \in C \bigcap S$, increase weight by $\beta$ \label{step3a}
\STATE Compute min-weight cut $C$ in $G_H$
\STATE $w_C \gets C$'s weight
\ENDWHILE
\IF {$2|C \bigcap S| < |C|$}
\STATE Construct attack for Problem using $C$
\ELSE
\STATE Declare no solution
\ENDIF
\ELSE
\WHILE {($w_C < \gamma) \&\& (|C \bigcap S^c| = 0 \text{~or~} 2|C \bigcap S^c| > |C|$)} \label{step2b}
\IF{$|C \bigcap S^c| = 0$}
\STATE Pick  $i \in C \bigcap S$, increase weight by $\infty$\label{step2_5b}
\ELSE
\STATE Pick $i \in C \bigcap S^c$, increase weight by $\beta$ \label{step3b}
\ENDIF
\STATE Compute min-weight cut $C$ in $G_H$
\STATE $w_C \gets C$'s weight
\ENDWHILE
\IF {$2|C \bigcap S^c| \leq |C|$}
\STATE Construct attack for Problem using $C$
\ELSE
\STATE Declare no solution
\ENDIF
\ENDIF
\end{algorithmic}
\end{algorithm}

\textbf{Working and Complexity:} We describe Algorithm $2$ with graph-cut constraint specified by Case A ($n^{S}_C < |C|/2$). The analysis for Case B follow in a similar way. The edge-weights of secure and insecure edges are specified by Problem I-A or II-A. Step~\ref{step1} computes the minimum weight cut $C$ in $G_H$ and checks if it satisfies the cut constraint in Case A (Step~\ref{step2a}). If the constraint is not satisfied, one secure edge is selected randomly in $C$ and its edge-weight is increased by $\beta$ (Step~\ref{step3a}). Here $\beta$'s value is taken as either $\infty$ or the secure edge-weight for Case A (insecure edge-weight for Case B). Following this, the minimum weight cut is recomputed and checked to see if the cut constraint is satisfied. This process is iterated until a feasible cut is obtained or the cut-weight grows beyond threshold $\gamma$, at which point the algorithm declares no solution.

We discuss the complexity for $\beta = \infty$ and Case A. Here, the algorithm computes a maximum of $|S|$ min-cut computations, one for each secure edge.  Since each min-cut can be computed in $O(|n||m|+|n|^2\log|n|)$ time \cite{mincut}, Algorithm $2$ has a worst-case computational complexity of $O(|S||n||m|+|S||n|^2\log|n|)$ for constraint specified by Case A.

It needs to be noted that the finding the existence of a feasible cut of Case A or B is NP-hard \cite{dekaISGT} and hence obtaining the optimal cut is NP-hard as well. Thus, Algorithm $2$ for optimal attack construction is approximate and may not return a solution for all system configurations. Determining the approximation gap of Algorithm $2$ will depend on approximations of the ratio-cut problem for feasibility and additionally on reducing the cut-size for optimality. In the next section, we present simulation results to justify the good performance of Algorithm $2$  in designing optimal `detectable' generalized attacks.

\section{Results on IEEE test systems}
\label{sec:results}
We discuss the performance of Algorithm $1$  and Algorithm $2$ in designing `hidden' and `detectable' generalized attacks by simulations on IEEE $14$-bus and $57$-bus test systems \cite{testsystem}. In each simulation run, we put flow measurements on all lines and phase angle measurements on $60\%$ (randomly selected) of the system buses. We vary the fraction of secure measurements in either system, and observe its effect on average cost of constructing data attacks as specified by Theorems \ref{hiddengeneralizeddesign}, \ref{detectablegeneralizeddesignI}, \ref{detectablegeneralizeddesignII} and \ref{detectablegeneralizeddesignIII}. We first consider Algorithm $1$ that gives the optimal `hidden' generalized attack as well as `detectable' generalized attack in Interval III. Here, the costs of jamming insecure and secure measurements are taken respectively as $.25$ and $.5$ relative to the cost of injecting data into an insecure measurement, respecting the inequality in Assumption $1$. Fig.~\ref{fig1} presents the trends in average costs of `hidden' injection, `detectable' injection and `hidden' generalized attacks for the IEEE $14$-bus and $57$-bus test systems for configurations where `hidden' injection attacks are feasible. It is clearly observed that adding jamming to the adversarial tools reduces the cost of `hidden' attacks greatly. In fact `hidden' generalized attacks are less expensive than `detectable' injection attacks which on average cost less than $50\%$ of the cost of `hidden' injection attacks \cite{dekaISGT}.

Next we consider Algorithm $2$ and use it to generate `detectable' generalized attacks in Intervals I and II (see Fig.~\ref{fig:regions}). For Intervals I and II specified in Fig.~\ref{fig:regions}, the relative costs of jamming an insecure measurement are respectively taken as $.6$ and $.25$ times the cost of data injection. The relative cost of jamming a secure measurement to that of data injection into an insecure measurement is taken as $.8$ in both intervals, as per Assumption $1$. To show the adversarial advantage of jamming secure measurements, we compare the average costs of `detectable' generalized (DG) attacks with that of `detectable' jamming (DJ) attacks in each case. Fig.~\ref{fig2} presents the average DG and DJ attack costs for the IEEE $14$-bus and $57$-bus test systems in cases with feasible `hidden' injection attacks. It can be observed that though jamming of secure measurements reduces the average attack costs, its effect is more pronounced in Interval I where cost of jamming an insecure measurement is higher. Similarly, Fig.~\ref{fig2} demonstrates the trends in average DG and DJ attack costs for the same systems, but by considering cases with feasible `detectable' injections attacks. Even in this case, the cost improvement in DG over DJ attacks is greater in Interval I.

Note that the rise in attack cost with increase in the fraction of secure measurements in the system is greater in Fig.~\ref{fig3} than in Fig.~\ref{fig1} and Fig.~\ref{fig2}. This disparity is due to the fact that in Figs.~\ref{fig1} and~\ref{fig2}, we only record attack costs for system configurations with feasible `hidden' injection attacks. As the number of such configurations decreases rapidly with increasing number of secure measurements, we end up averaging over fewer configurations leading to lower recorded average attack costs. The number of feasible configurations for `detectable' injection attacks does not decrease as rapidly, hence Fig.~\ref{fig3} has cost curves with steeper slopes in general.
\begin{figure}[ht]
\centering
\includegraphics[width=0.5\textwidth]{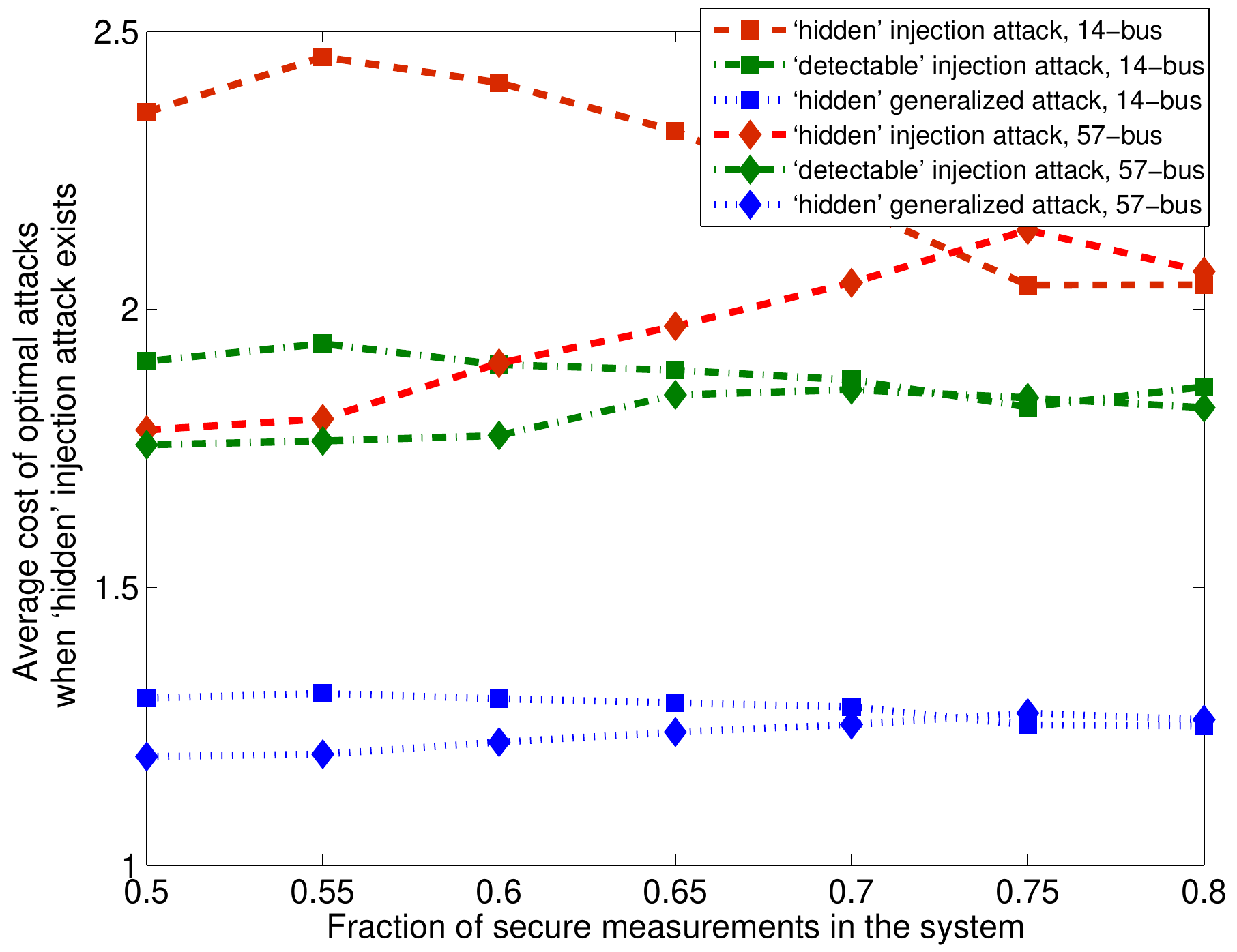}
\caption{Average cost of `hidden' injection, `detectable' injection and `hidden' generalized attacks (when `hidden' injection attack exists) produced by Algorithm $1$ on the IEEE $14$ and $57$ bus test systems with flow measurements on all lines, phasor measurements on $60\%$ of the buses and protection on a fraction of measurements selected randomly. The cost of data injection ($p_I$) is taken as $1$. The costs of jamming an insecure measurement ($p_J^{S^c}$) and a secure measurement ($p_J^S$) are taken as $.25$ and $.5$ respectively.}
\label{fig1}
\end{figure}
\begin{figure}
\centering
\includegraphics[width=0.5\textwidth]{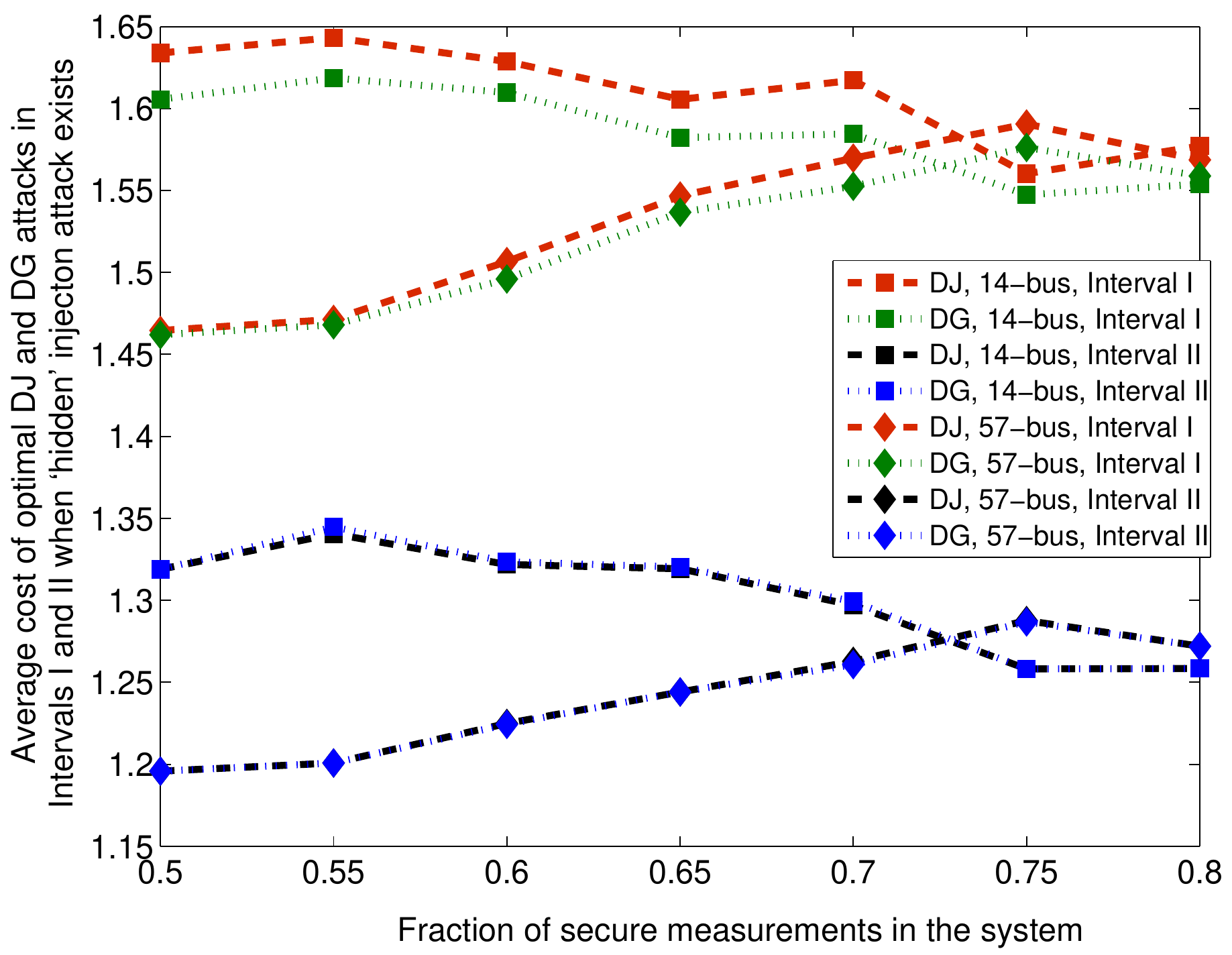}
\caption{Average cost of `detectable' generalized (DG) and `detectable' jamming (DJ) attacks (when `hidden' injection attack exists) in Cost Intervals I and II, produced by Algorithm $2$ (with finite $\beta$) on the IEEE $14$ and $57$ bus test systems with flow measurements on all lines, phasor measurements on $60\%$ of the buses and protection on a fraction of measurements selected randomly. In Interval I and II, the costs of jamming an insecure measurement are taken as $.6$ and $.25$ respectively. The costs of jamming a secure measurement and data injection are taken as $.8$ and $1$ respectively in both intervals.}
\label{fig2}
\end{figure}
\begin{figure}
\centering
\includegraphics[width=0.5\textwidth]{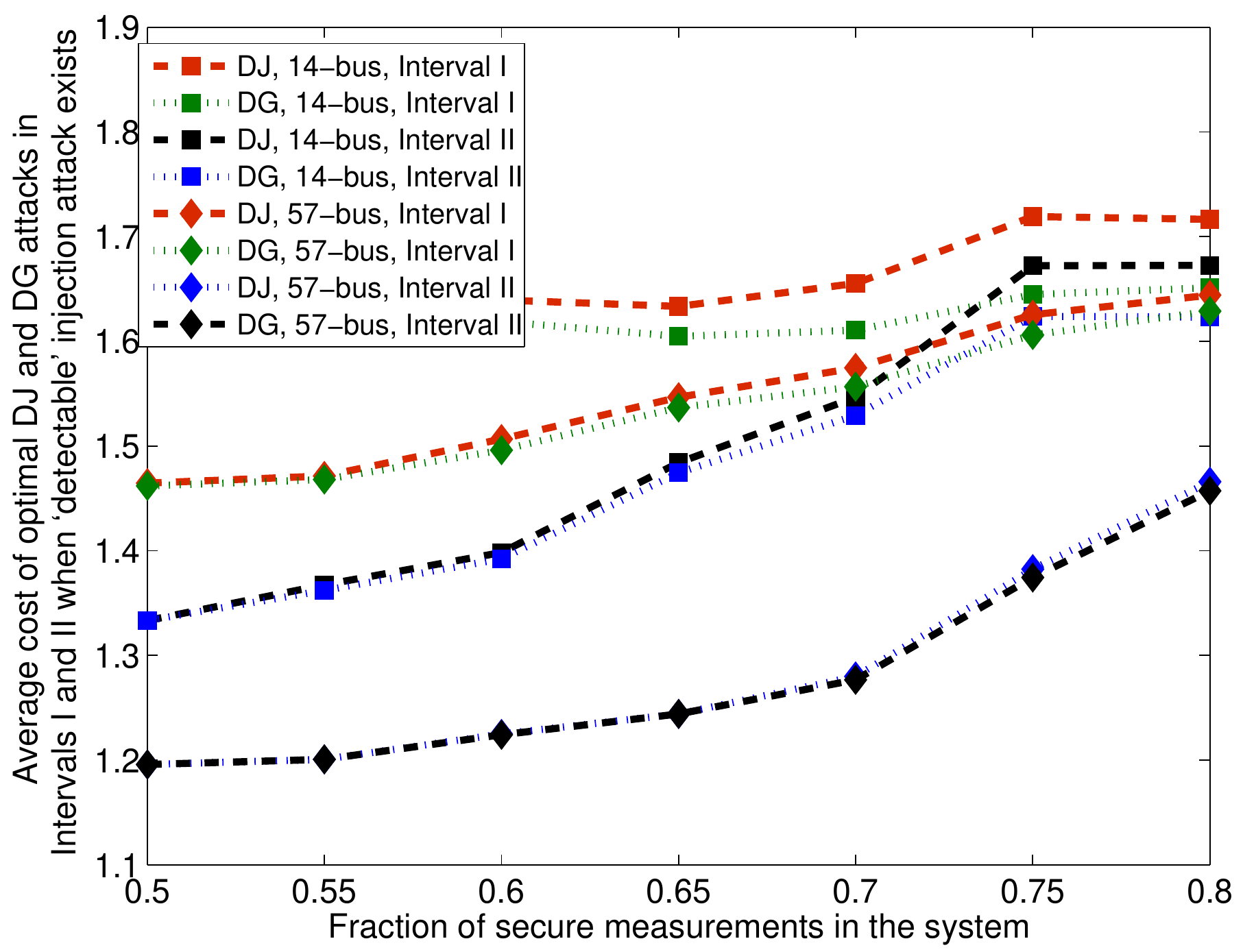}
\caption{Average cost of `detectable' generalized (DG) and `detectable' jamming (DJ) attacks (when `detectable' jamming attack exists) in Cost Intervals I and II produced by Algorithm $2$ (with finite $\beta$) on the IEEE $14$ and $57$ bus test systems with flow measurements on all lines, phasor measurements on $60\%$ of the buses and protection on a fraction of measurements selected randomly. In Interval I and II, the costs of jamming an insecure measurement are taken as $.6$ and $.25$ respectively. The costs of jamming a secure measurement and data injection are taken as $.8$ and $1$ respectively in both intervals.}
\label{fig3}
\end{figure}

\section{Conclusion}
\label{sec:conclusion}
We introduce `generalized' data attacks on state estimation in this paper. In our attack framework, an adversary uses three tools with distinct costs: jamming of encrypted (secure) measurements, data injection and jamming of insecure measurements to optimize the cost and expand the scope of traditional data attacks in literature. We consider both `hidden' and `detectable' data attacks and present novel graph-cut based formulations for construction of optimal generalized attacks of each type. We show that the optimal `hidden' attack with adversarial jamming is given by the minimum weight graph-cut where the edge-weights for secure and insecure measurements are based on the costs of jamming and data injection in the system. We prove that the optimal `hidden' attack with jamming is exactly constructed using a polynomial time min-cut based algorithm. For `detectable' attacks, we show that the entire range of relative costs for data injection and jamming of secure and insecure measurements can be divided into three separate intervals, each with distinct `constrained graph-cut' based optimal attack construction. We present approximate algorithms that use iterative min-cut computations to determine the optimal `detectable' attack in each interval. Due to the ability to jam secure measurements, our generalized framework has very relaxed constraints on attack feasibility compared to traditional models. This reduces the cost of `hidden' and `detectable' attacks as well as increases adversarial immunity against grid security. Specifically, we show that our generalized data attacks are even feasible for systems with a single insecure measurement and hence are not prevented by adding new secure measurements. We present simulation results of our proposed attack framework on IEEE test cases for different costs of adversarial tools and discuss the performance of our algorithms. Jamming of secure measurements indeed severely weakens grid security by reducing attack cost and expanding attack feasibility significantly over that of traditional data attacks. Techniques to efficiently prevent generalized attacks by improving state estimation and theoretical analysis of the performance of our designed approximate algorithms are directions of our future work in this domain.


\begin{thebibliography}{1}
\bibitem{pmu1}
A. G. Phadke, ``Synchronized phasor measurements in power systems", {\em IEEE Comput. Appl. Power}, vol. 6, 1993.
\bibitem{wallstreet}
S. Gorman, ``Electricity grid in U.S. penetrated by spies", {\em Wall St. J.}, 2009.
\bibitem{dragonfly}
http://www.nytimes.com/2014/07/01/technology/energy-sector-faces-attacks-from-hackers-in-russia.html
\bibitem{aurora}
J. Meserve, ``Staged cyber attack reveals vulnerability in power grid", CNN, 2007. Available: http://www.cnn.com/2007/US/ 09/26/power.at.risk/index.html.
\bibitem{todd}
Shepard, D. P., Humphreys, T. E., and Fansler, A. A., ``Evaulation of the Vulnerability of Phasor Measurement Units to GPS Spoofing", {\em International Journal of Critical Infrastructure Protection}, 2012.
\bibitem{hidden}
Y. Liu, P. Ning, and M. K. Reiter, ``False data injection attacks against state estimation in electric power grids", {\em Proc. ACM Conf. Comput. Commun. Security}, 2009.
\bibitem{poor}
T. Kim and V. Poor, ``Strategic Protection Against Data Injection Attacks on Power Grids", {\em IEEE Trans. Smart Grid}, vol. 2, no. 2, 2011.
\bibitem{sou}
O. Vukovic, K. C. Sou, G. Dan, and H. Sandberg, ``Network-aware mitigation of data integrity attack on power system state estimation", {\em IEEE Journal on Selected Areas in Communications}, vol. 30, no. 6, 2012.
\bibitem{deka}
D. Deka, R. Baldick, and S. Vishwanath, ``Optimal Hidden SCADA Attacks on Power Grid: A Graph Theoretic Approach", {\em ICNC}, 2014.
\bibitem{deka1}
D. Deka, R. Baldick, and S. Vishwanath, ``Data Attack on Strategic Buses in the Power Grid: Design and Protection", {\em IEEE PES General Meeting}, 2014.
\bibitem{thomas}
O. Kosut, L. Jia, R. J. Thomas, and L. Tong, ``Limiting false data attacks on power system state estimation", {\em Proc. Conf. Inf. Sci. Syst.}, 2010.
\bibitem{frame}
J. Kim, L. Tong, and R. J. Thomas, ``Data Framing Attack on State Estimation with Unknown Network Parameters", {\em Asilomar Conference on Signals, Syst., and Computers}, 2013.
\bibitem{dekaISGT}
D. Deka, R. Baldick, and S. Vishwanath, ``Data Attacks on the Power Grid DESPITE Detection", {\em IEEE PES Innovative Smart Grid Technologies}, 2015. Available at: http://arxiv.org/abs/1505.01881
\bibitem{florian}
F. Pasqualetti, F. Dorfler, and F. Bullo, ``Attack detection and identification in cyber-physical systems", {\em IEEE Transactions on Automatic Control}, vol. 58, 2013.
\bibitem{ddos}
L. Shichao, L. P. Xiaoping, and S. E. Abdulmotaleb, ``Denial-ofservice (dos) attacks on load frequency control in smart grids", {\em IEEE PES Innovative Smart Grid Technologies}, 2013.
\bibitem{wsj}
R. Smith, ``Assault on California Power Station Raises Alarm on Potential for Terrorism". Available at: http://www.wsj.com/articles/SB 10001424052702304851104579359141941621778
\bibitem{infocom}
D. Deka, R. Baldick, and S. Vishwanath, ``Attacking Power Grids with Secure Meters: The Case for using Breakers and Jammers", {\em IEEE Infocom CCSES Workshop}, 2014.
\bibitem{dekaPESGM2015}
D. Deka, R. Baldick, and S. Vishwanath, ``One Breaker is Enough: Hidden Topology Attacks on Power Grids", {\em IEEE PES General Meeting}, 2015. Available at: http://arxiv.org/abs/1506.04303
\bibitem{dekasmartgridcomm2015}
D. Deka, R. Baldick, and S. Vishwanath, ``Optimal Data Attacks on Power Grids:
Leveraging Detection \& Measurement Jamming", {\em IEEE Smartgridcomm}, 2015 (accepted). Available at: http://arxiv.org/abs/1506.04541
\bibitem{testsystem}
R. Christie, ``Power system test archive", Available: http://www.ee.washington.edu/research/pstca.
\bibitem{monticelli}
A. Monticelli, ``State estimation in electric power systems: a generalized approach", {\em Kluwer Academic Publishers}, 1999.
\bibitem{bookDC}
A. Abur and A. G. Exposito, ``Power System State Estimation: Theory and Implementation", {\em CRC}, 2000.
\bibitem{price}
L. Xie, Y. Mo, and B. Sinopoli, ``False data injection attacks in electricity markets", {\em IEEE SmartGridComm}, 2010.
\bibitem{ratio}
M. R. Garey and D. S. Johnson, ``Computers and Intractability: A Guide to the Theory of NP-Completeness", {\em W. H. Freeman}, 1979.
\bibitem{boyd}
S. Boyd and L. Vandenberghe, ``Convex Optimization", {\em Cambridge University Press}, 2004.
\bibitem{SDP}
M. X. Goemans and D. P. Williamson, ``Improved approximation algorithms for maximum
cut and satisfiability problems", {\em Journal of the ACM}, vol. 42, 1995.
\bibitem{maxflow}
A. V. Goldberg and R. E. Tarjan, ``A new approach to the maximum-flow problem", {\em Journal of the ACM}, vol. 35, 1988.
\bibitem{mincut}
M. Stoer and F. Wagner, ``A simple min-cut algorithm", {\em Journal of the ACM}, vol. 44, 1997.
\end{thebibliography}
\end{document}